\documentclass[a4paper]{article}   
\usepackage{amsmath}
\usepackage{amssymb}
\usepackage{amsthm}
\usepackage{color}
\newtheorem{theorem}{Theorem}
\newtheorem{definition}{Definition}
\newtheorem{corollary}{Corollary}

\usepackage{cite}
\date{}

\def\bea{\begin{eqnarray}}
\def\eea{\end{eqnarray}}
\def\be{\begin{equation}}
\def\ee{\end{equation}}
\def\bes{\begin{equation*}}
\def\ees{\end{equation*}}
\def\bs{\begin{split}}
\def\es{\end{split}}
\def\beas{\begin{eqnarray*}}
\def\eeas{\end{eqnarray*}}

\begin{document}
\title{Equivalence Groups and Differential Invariants for (2+1) dimensional Nonlinear Diffusion Equation}
\author{Saadet  \"Ozer }
\maketitle {} Istanbul Technical University, Faculty of Science and Letters,\\
Department of Mathematics Engineering, 34469 Maslak Istanbul
Turkey;\\
email: saadet.ozer@itu.edu.tr\\
\begin {abstract}
(2+1) dimensional
diffusion equation is considered within the framework of equivalence
transformations. Generators for the group are obtained and admissible transformations
between linear and nonlinear equations are examined. It is  shown that  transformations between  linear and  nonlinear equations are possible provided that the generators of independent variables depend on  the dependent variable.  Exact solutions for some nonlinear equations are obtained.
Differential invariants related to
the transformation groups are  investigated and the results are compared with the direct integration method.\\
\end{abstract}
{\bf keywords:}
Lie Group Application, Equivalence Groups , Exact Solution,  Nonlinear Diffusion Equation,  Differential Invariants

\section{Introduction}
Differential equations containing some arbitrary functions or
parameters represent actually  family of equations of the same
structure. Almost all field equations of classical continuum physics
possess this property related to the behaviour of
different materials.  In dealing with such family of differential
equations, Lie symmetry analysis provides some powerful algorithmic
methods for determination of invariant solutions, conserved
quantities and construction of maps between
differential equations of the same family that turns out to be
equivalent \cite{csuhubi2005explicit,oliveri2012equivalence,ozer2003equivalence}.
 To examine such problems, it is convenient considering equivalence transformation groups that preserve
 the structure of the family of differential equations but may change the form of the constitutive functions, parameters
 when appropriate transformations are available.\par
 The first systematic treatment that the usual Lie's infinitesimal invariance approach could be
employed in order to construct equivalence groups was formulated by
Ovsiannikov ~\cite{ovsiannikov2014group}. Then several well-developed methods have been used to construct equivalence groups. The general theory of determining
transformation groups and algorithms  can be found in the references
\cite{ovsiannikov2014group,olver2000applications,
ibragimov1999elementary,suhubi2013exterior}.
\par In the present text we shall examine the (2+1) dimensional diffusion equation. Nonlinear members of the family of diffusion equations have significant importance in many areas in applied sciences. A great number of members  have been
used to model many physical phenomena in Mathematical Physics, Mathematical Biology etc. and  they have been widely studied not only by means of numerical, asymptotic analysis but also
as application to Lie Group Analysis since
Lie \cite{lie1897uber}. 
The simple nonlinear
heat equation; $u_t=[A(u)u_x]_x$ was first examined by Ovsiannikov
\cite{ovsiannikov1959group} within the frame-work of Lie's symmetry classification. The complete classification  and form-preserving point transformations for the inhomogenous one dimensional nonlinear diffusion are obtained in \cite{sophocleous2003symmetries}. The conditions for reduction the more general diffusion type equations to  the one dimensional  heat equation are also examined in \cite{johnpillai2001singular}. 
Equivalence transformations of  linear diffusion equations into  nonlinear equations for some classes have   been considered in Lisle's PhD thesis \cite{lisle1992equivalence} widely. Constructing the exact analytical solutions for some specific problems related to the nonlinear diffusion equation have recently examined by  \cite{ivanova2010lie,torrisi2007class,popovych2004new}. Torrisi et al.  in their paper \cite{torrisi2007class} have also studied the developments of bacterial colonies as application to  equivalence groups. Bruz{\'o}n et al. have derived maps between nonlinear dispersive equations in their detailed work \cite{bruzon2012some}.
 \par
 The aim of the present work is to study equivalence transformations
for a general family of (2+1) dimensional diffusion
equation. We investigate the structure of the transformation group generators which lead to map linear and
nonlinear members. 
\par
For the convenience of the reader to follow, in section 2,  we
have obtained the generators of the group of equivalence transformations and   determined the  structure   for admissible transformations. Theorems in the section  state the conditions of  the appropriate equivalence  transformations   between linear and nonlinear members . In the next
section, we have considered some subgroups of the general equivalence groups and  by choosing some specific forms we
are able to obtain some classes of nonlinear equations which are
equivalent to linear ones. We have also investigated the classes of  nonlinear diffusion equations that are mapped onto the classical heat equation. Exact solutions for those nonlinear equations are also obtained. In section 4 we have also examined
differential invariants for the subgroups and discussed the results with the direct integration method.
\section{Equivalence Transformations}
In the present paper we shall investigate the equivalence group a general family of (2+1) dimensional diffusion equation
\be
\label{main}
u_t=f(x,y,t,u,u_x,u_y)_x +g(x,y,t,u,u_x,u_y)_y
\ee
 which represents a great variety  of linear and nonlinear equations. Here $u$ is the dependent variable of the independent variables $x,y$ and $t$. Here $f$ and $g$
are smooth nonconstant functions of their variables and subscripts
denote the partial derivatives with respect to the corresponding variables.
\begin{definition}
With $n$ independent variables $x_i$, $N$ dependent variables $u_\alpha$ and $m$ smooth functions $\phi_k$ of independent, dependent variables and their derivatives
$$
{\cal F}(x_i,u_{\alpha(p)},\phi_{k(q)}(x_i,u_{\alpha (p)}))=0
$$
is called  a family of differential equations. Here
 $ i=1,2...n,\ \alpha=1,2,...N,\ k=1,2,...m$ and $u_{\alpha (p)}$ include both the tuple of dependent variables $u=(u_1,u_2,...,u_N)$ as well as all the derivatives of $u$ with respect to $x_i$'s up to order $N$. By $\phi_{k(q)}$ we denote the smooth functions $\phi_k$ and  the partial derivatives with respect to both $x_i$'s and $u_{\alpha (p)}$'s.
\end{definition}
 \begin{definition}
For a given differential equation of the family
$$
{\cal F}(x_i,u_{\alpha(p)},\phi_{k(q)}(x_i,u_{\alpha (p)}))=0
$$
the equivalence group $\cal{E}$  is the group of smooth transformations of independent, dependent variables, their derivatives  and smooth functions preserving the structure of the differential equation but transforms it into
$${\cal F}(\bar x_i,\bar u_{\alpha(\bar p)},\bar\phi_{k(\bar q)}(\bar x_i,\bar u_{\alpha (\bar p)}))=0. $$
 \end{definition}
 More precisely, equivalence transformations associated with the (2+1) dimensional most general
diffusion equation \eqref{main} transform the equation into
\begin{align*}
& u_t-f(x,y,t,u,u_x,u_y)_x -g(x,y,t,u,u_x,u_y)_y=0\ \longrightarrow \\
& \bar
u_{\bar t}-\bar f(\bar x,\bar y,\bar t,\bar u,\bar u_{\bar x},\bar
u_{\bar y})_{\bar x} -\bar g(\bar x,\bar y,\bar t,\bar u,\bar
u_{\bar x},\bar u_{\bar y})_{\bar y}=0.
\end{align*}
where $\bar {(.)}$ represents the transformed variables and functions.\par
Let $\mathcal{M}=\mathcal{N}\times \mathbb R$  be a (2+1) dimensional manifold with a local coordinate system $\mathbf{x}=( x_{i}) =( x,y,t)$ which we shall call as the
space of independent variables. Consider a trivial bundle structure $(
\mathcal{K},\pi ,\mathcal{M}) $ with fibers are the real line $\mathbb{R}$. Here $\mathcal{M}$ is the base manifold and $\mathcal{K}$, called the
graph space is globally in form of a product manifold $\mathcal{M}\times \mathbb{R}$. We equip the four dimensional graph space $\mathcal{K}$ with
the local coordinates $( \mathbf{x},u) =( x,y,t,u) $.
\par
 A vector field on the graph space $\mathcal{K}$ is a section of its tangent
bundle and locally in form
 \be
V=\xi^1\frac{\partial}{\partial x}+\xi^2\frac{\partial}{\partial y}+\xi^3\frac{\partial}{\partial t}+\eta\frac{\partial}{\partial u}
\ee
 where $\xi^i\ (i=1,2,3)$ and $\eta$ are coordinate functions on $\mathcal{K}$.
\par
 In order to construct the equivalence groups $\mathcal{E}$ for the equation \eqref{main}, first we  extend the graph space $\mathcal{K}$  by adding the auxiliary variables and the ones representing the functional dependencies of the smooth functions $f$ and $g$
 \be
 \label{manifold}
 \tilde{\mathcal{K}}=\{x,y,t,u,f,g,u_x,u_y,u_t,f_x,f_y,f_t,f_u,g_x,g_y,g_t,g_u,f_{u_x},f_{u_y},g_{u_x},g_{u_y}\}.
 \ee
The prolongation vector $\tilde V$ over the extended manifold covered by $\tilde K$   can be written as
\bea \label{prolonged}\bs
 &\tilde{V} =V+\mu^1 \frac{\partial }{{\partial f }} +\mu^2  \frac{\partial }{{\partial g }} +\zeta^1
   \frac{\partial }{{\partial u_x }} +\zeta^2  \frac{\partial }{{\partial u_y   }}
   + \zeta^3\frac{\partial }{{\partial u_t }} + \sum_{j=1}^4\mu_j^1 \frac{\partial }
   {{\partial f_j }}  \\
 &\qquad+\sum_{j=1}^4\mu_j^2 \frac{\partial }
   {{\partial g_j }}+\sum_{j=1}^3\nu_j^1 \frac{\partial }
   {{\partial f_{u_j} }}+\sum_{j=1}^3\nu_j^2 \frac{\partial }
   {{\partial g_{u_j} }}
 \end{split}
\eea
where in the last four summations $j=1,..4$ represent $x,y,t$ and  $u$ and all coefficient functions are smooth functions of the coordinates of the extended manifold.  
\begin{theorem}[\cite{suhubi2013exterior}]
Let a vector field on an $n$ dimensional differentiable manifold $M$ be given by $$V(p)=v^i({\bf x})\frac{\partial}{\partial x^i}, \quad p=\varphi^{(-1)}({\bf x}),\quad i=1,...,n$$
where $(U,\varphi)$ is the chart to which $p\in M$ belongs. A curve $\gamma$ is an integral curve of the vector field $V$ iff the coordinate functions $x^i(t)$ are solutions of the following system of local ordinary differential equations in $\mathbb R^n$
$$
\frac{d x^i}{dt}=v^i({\bf x}(t))
$$
\end{theorem}

Precisely,  the equivalence transformations  can be determined by
solving the following system of autonomous ordinary differential equations on the extended manifold \eqref{manifold} \bea
\label{system}
\begin{split}
 &\frac{{d\bar x}}{{d\epsilon }} = \xi^1 (\bar x,\bar y,\bar t,\bar u ),
 \qquad\frac{{d\bar y}}{{d\epsilon }} = \xi^2 (\bar x,\bar y,\bar t,\bar u ),
 \qquad \frac{{d\bar t}}{{d\epsilon }} =\xi^3 (\bar x,\bar y,\bar t,\bar u ),\\
 &\frac{{d\bar u}}{{d\epsilon }} = \eta(\bar x, \bar y,\bar t,\bar u ),
 \qquad \frac{{d\bar f }}{{d\epsilon }} = \mu^1,
 \qquad\frac{{d\bar g }}{{d\epsilon }} = \mu^2 \\
 \end{split}
\eea under the initial conditions
 \be
 \label{initial}
 \bar x (0) = x ,\quad \bar
y(0)=y,\quad \bar t(0)=t,\quad \bar u(0) = u,\quad \bar f (0) = f
,\quad\bar g (0) = g \ee where $\mu^1$ and $\mu^2$ depend on $(\bar x,\bar y,\bar t,\bar u,\bar u_{\bar x}, \bar u_{\bar y},\bar u_{\bar t},\bar f,\bar g)$.
\\
\par 
Coefficients of the prolonged vector field \eqref{prolonged}, namely the infinitesimal generators for the equivalence group can be evaluated by the very well-known prolongation formula (for detailed information, theorems and detailed applications see \cite{olver2000applications,akhatov1991nonlocal}). For more details we  refer the reader the papers \cite{ibragimov1991preliminary,
torrisi1998equivalence,romano1999application} and references therein which are concerned  with equivalence transformations.
\bea
\label{gen1}
\bs
& \zeta^1=D_x(\eta)-u_x D_x(\xi^1)-u_y D_x(\xi^2)-u_t D_x(\xi^3),\\
& \zeta^2=D_y(\eta)-u_x D_y(\xi^1)-u_y D_y(\xi^2)-u_t D_y(\xi^3),\\
& \zeta^3=D_t(\eta)-u_x D_t(\xi^1)-u_y D_t(\xi^2)-u_t D_t(\xi^3)
 \end{split}
\eea
where $D_x,D_y,D_t$ denote the total derivatives with respect to their parameters: $D_i=\frac{\partial}{\partial x_i}
+u_{x_i}\frac{\partial}{\partial u}$. And
\bea
\label{gen2}\bs
&\mu_j^1=\tilde{D}_j(\mu^1)-\sum_{i=1}^3f_i\tilde{D}_j(\xi^i)-f_u\tilde{D}_j(\eta)-\sum_{i=1}^3f_{u_i}\tilde{D}_j(\zeta^i),\\
&\mu_j^2=\tilde{D}_j(\mu^2)-\sum_{i=1}^3g_i\tilde{D}_j(\xi^i)-g_u\tilde{D}_j(\eta)-\sum_{i=1}^3g_{u_i}\tilde{D}_j(\zeta^i),\\
& \nu_j^1=\tilde{D}_{u_j}(\mu^1)-\sum_{i=1}^3 f_{u_i} \tilde{D}_{u_j}(\zeta^i),\qquad\\
&\nu_j^2=\tilde{D}_{u_j}(\mu^2)-\sum_{i=1}^3 g_{u_i} \tilde{D}_{u_j}(\zeta^i)\qquad
 \end{split}
\eea
where  $\tilde{D}_j=\frac{\partial}{\partial x_j}$ and $\tilde{D}_{u_j}=\frac{\partial}{\partial {u_j}}$.
\par
These expressions  do not impose any restriction on functional dependencies of the smooth functions $f$ and $g$ in the main equation \eqref{main}. If some variables do not appear in the coordinate cover of
the extended manifold \eqref{manifold}, due to a particular structure of the given differential  equation, that might entail some restrictions on the extended vector field components because the
corresponding  components must then be set to zero. Note that  in equation (1) the free parameters $f$ and $g$ do not depend on $u_t$. Thus their corresponding components in \eqref{prolonged} must vanish:
\be
\label{restriction}
\mu_3^1=\mu_3^2=0.
\ee
\begin{theorem}
A nonlinear (2+1) dimensional diffusion equation can be mapped onto a linear equation by a point equivalence transformation, if and only if it is in the following form:
 \bes 
 \bar x=\phi(x,y,t,u)\quad  {\text {and/or} }\quad  \bar y=\psi(x,y,t,u)
 \ees where $\phi, \psi
\in C^2$ and $\frac{\partial \phi}{\partial u}\neq 0$,
$\frac{\partial \psi}{\partial u}\neq 0$.
\end{theorem}
\begin{proof}
To generate the transformations of the  equivalence group for the diffusion equation given by \eqref{main}, we shall apply the restrictions  \eqref{restriction}   to the given formulas  \eqref{gen1} and \eqref{gen2}.
Then we have
 \bea
 \label{generators}
\begin{split}
&\xi^1=\xi^1(x,y,t,u), \quad \xi^2=\xi^2(x,y,t,u),\quad  \xi^3=\xi^3(t),\quad \eta=\eta(x,y,t,u),\\
&\zeta^1= \eta_x+(\eta_u+\xi^1_x)u_{x}+\xi^1_u (u_x)^2+\xi^2_x u_y+\xi^2_u u_x u_y,\\
&\zeta^2= \eta_y+(\eta_u+\xi^2_y)u_y+\xi^1_y u_x+\xi^1_u u_x u_y+\xi^2_u (u_y)^2,\\
&\zeta^3= \eta_t+(\eta_u+\dot{\xi^3})u_t+\xi^1_t u_x+\xi^1_u u_x u_t+\xi^2_tu_y+\xi^2_u u_y u_t,\\
&\mu^1=(\eta_u+\dot{\xi^3}-\xi^1_x+\xi^2_uu_y)f-(\xi^1_y+\xi^2_u u_y)g+\gamma u_y+\kappa^1,\\
&\mu^2=(\eta_u+\dot{\xi^3}-\xi^2_y+\xi^1_u u_x)g-(\xi^2_x+\xi^2_u u_x)f-\gamma u_x+\kappa^2
\end{split}
\eea
and
\be
\label{result}
\eta_u=\xi^1_x+\xi^2_y-\dot{\xi^3}+s(t)\ee
 where $s(t)$ is an arbitrary continuous function, $\gamma,\ \kappa^1$ and $\kappa^2$ depend on the independent and dependent variables and  satisfy the following relations $$ \xi^1_t=-\gamma_y+\kappa^1_u,\qquad
\xi^2_t=\gamma_x+\kappa^2_u,\qquad \eta_t=\kappa^1_x+\kappa^2_y.$$
Equation  \eqref{result} points out that  $\eta$ can not depend on $u$ nonlinearly, unless $\xi^1$
and/or $\xi^2$  depend on $u$. Thus to have any transformation between linear and nonlinear equations  $\xi^1$ and/or $\xi^2$ must involve the dependent variable $u$  which meant to be the transformed variables $\bar x$ and /or $\bar y$ involve $u$:
\bes
\xi^1=\xi^1(x,y,t,u), \ {\text{and/or}} \ \xi^2=\xi^2(x,y,t,u),\quad \frac{\partial \xi^i}{\partial u}\neq 0, \quad i=1,2.
\ees
  Thus transformed variables are obtained via the solution of the set of ordinary differential equations \eqref{system} as
  \bes
   \bar x=\phi(x,y,t,u)\quad  {\text {and/or} }\quad  \bar y=\psi(x,y,t,u).
 \ees
   
\end{proof}
\par
Maps between nonlinear and linear equations might  sometimes be generated by some  nonlinear  dependencies on only the dependent variable;  like the very well-known wave equation (see \cite{torrisi2004linearization,sophocleous2008differential,
ibragimov2010quasi}).  
The following theorem points out such transformations are not admissible for diffusion equation.
\begin{theorem}
(2+1) dimensional diffusion equation does not admit the following type of transformation:
\bes
\bar x=\psi_1(x,y,t),\ \bar y=\psi_2(x,y,t),\ \bar t=\psi_3(t),\ \bar u=\psi_4(x,y,t,u)\ees
where $\psi_4(x,y,t,u)$ is nonlinear in $u$.
\begin{proof}
One can easily see  from the equation \eqref{result}, unless $\xi^1$ or $\xi^2$ depend on $u$, $\eta$ can not involve $u$. Thus $\bar u$ can not be nonlinear in $u$ when $\bar x$ or $\bar y$ does not involve $u$. 
\end{proof}
\end{theorem}
\begin{corollary}
Any integrable transformation related to the generators; $\xi^1=\xi^1(x,y,t),\ \xi^2=\xi^2(x,y,t),\ \xi^3=\xi^3(t),\ \eta=h(x,y,t)u$  map a linear
equation onto another linear equation
  with various different coefficient functions.
\end{corollary}
Readers may see that equations \eqref{generators} generate all set of admissible equivalence transformations, related to the general (2+1)dimensional diffusion equation  \eqref{main}. The equivalence transformations for some particular members of the family of \eqref{main} can be evaluated from these by integrating the system of equations \eqref{system}.

\section{Applications to exact solutions of nonlinear diffusion equations}
Theorem 2 also refers that the admissible equivalence transformations  related to the diffusion equation \eqref{main} may involve some arbitrary  functions.  As a consequence of this result, in addition to the maps between  single linear and  nonlinear equations,  one can say that transformations between  a  single linear equation and a class of nonlinear equations can also be constructed. That is to say we can set such transformations  which map a single linear equation into a particular family of nonlinear equations or vice versa.
 $$
 f_x+g_y-u_t=0 \longleftrightarrow \bar f_{k_{ \bar x}}+\bar g_{k_{ \bar y}}-\bar u_{k_{ \bar t}}=0
 $$
 where the subscript $k$ denotes  the class of equations.
\par 
 In this section  we shall investigate some applications to that circumstance and study such maps  by considering  the transformations involving some arbitrary differentiable functions. Hereby solutions of some particular family of nonlinear equations can be obtained from an appropriate  linear equation.

  \subsection{On the subgroup: $\xi^1=m(u),\ \xi^2=h(u), \ \xi^3=\eta=0$  }
Here we shall examine a subgroup of the admissible equivalence groups generated by the infinitesimal generators: $\xi^1=m(u),\ \xi^2=h(u), \ \xi^3=\eta=0$  where $m(u)$ and $h(u)$ are arbitrary differentiable continuous functions.
 Then the prolonged vector field \eqref{prolonged} can be written via \eqref{generators} as follows
 \be
 \tilde{V}=m(u)\frac{\partial}{\partial x}+h(u)\frac{\partial}{\partial y}+(h'(u)f-m'(u)g)u_y\frac{\partial}{\partial f}+(m'(u)g-h'(u)f)u_x\frac{\partial}{\partial g}+\cdots
 \ee
 By integrating the system of equations \eqref{system} under the initial conditions \eqref{initial} we have the following class of transformations for the subgroup
 \be
 \label{case1}
 \begin{split}
 \bar x&=x-\epsilon m(u),\quad \bar y=y-\epsilon h(u),\quad \bar t=t,\quad \bar u=u,\\
 \bar u_x&=\frac{u_x}{1-\epsilon(u_x m'(u)+u_yh'(u))},\quad \bar u_y=\frac{u_y}{1-\epsilon(u_xm'(u)+u_yh'(u))},\\
  \bar u_t&=
\frac{u_t}{1-\epsilon(u_xm'(u)+u_yh'(u))},\\
 \bar f&=\frac{(1- \epsilon u_x  m'(u))f-\epsilon u_y  m'(u)g}{1-\epsilon(u_xm'(u)+u_yh'(u))},\quad
 \bar g=\frac{(1-\epsilon u_y h'(u))g-\epsilon u_xh'(u)f}{1-\epsilon(u_xm'(u)+u_yh'(u))}.
 \end{split}
 \ee
Note that here $\epsilon$  is the group parameter.  By substituting $u_x$ and $ u_y$ in terms of the transformed variables,
 $\bar f$ and $\bar g$   can now be written as
\be
\label{transformed1}
 \bar f=\left(1+\epsilon \bar u_{\bar y} h'(\bar u)\right)\tilde f-\epsilon \bar u_{\bar y}m'(\bar u)\tilde g,\quad
  \bar g= \left(1+\epsilon \bar u_{\bar x} m'(\bar u)\right)\tilde g-\epsilon \bar u_{\bar x}h'(\bar u)\tilde f
 \ee
 where $\tilde f(\bar x,\bar y,\bar t,\bar u,\bar u_{\bar x},\bar
u_{\bar y})$ and $\tilde g(\bar x,\bar y,\bar t,\bar u,\bar u_{\bar
x},\bar u_{\bar y})$  represent $f$ and $g$ in terms of the
transformed variables.\\
\par 
It is clear that for every different choice of $m(u),\ h(u)$ and linear functions $f$ and
 $g$ in the dependent variable $u$ and its derivatives, the transformations \eqref{transformed1}  map a single linear differential equation $f_x+g_y-u_t=0$  onto a class of  nonlinear equations of the form 
  $ \bar f_{k_{ \bar x}}+\bar g_{k_{ \bar y}}-\bar u_{k_{ \bar t}}=0$, and  a solution of the linear equation $$\phi(x,y,t,u)=0$$ generates  solutions to the corresponding nonlinear
 equations as
 $$\phi_k(\bar x+\epsilon m(\bar u),\bar y+\epsilon h(\bar u),\bar t,\bar u)=0.$$

 {\bf{Example 1:}} As a particular example, we search the nonlinear diffusion equations of the class (1) that can be mapped onto the very well known heat equation. The (2+1)  dimensional  heat equation can be represented as a member of \eqref{main} by choosing   $ f=u_x,\ g=u_y$;
   \be
\label{heat}
u_{xx}+u_{yy}= u_t.
\ee
The transformed form of $f=u_x$ and $g=u_y$ can be obtained  from \eqref{transformed1} as
 \bes
  \bar f=\frac{(1+\epsilon h'(\bar u)\bar u_{\bar y} )\bar u_{\bar x}-\epsilon m'(\bar u)\bar u_{\bar y}^2}{1+\epsilon h'(\bar u)\bar u_{\bar y} +\epsilon m'(\bar u)\bar u_{\bar x}},\quad
  \bar g=\frac{(1+\epsilon m'(\bar u)\bar u_{\bar x} )\bar u_{\bar y}-\epsilon h'(\bar u)\bar u_{\bar x}^2}{1+\epsilon h'(\bar u)\bar u_{\bar y} +\epsilon m'(\bar u)\bar u_{\bar x}}
 \ees
which generate the following class of nonlinear equations
 \bea
 \label{nonlinear1}
  A \ \bar u_{\bar x\bar x}+ B \ \bar u_{\bar x \bar y} +C \ \bar u_{\bar y\bar y} + D = E \ \bar u_{\bar t}
 \eea
 where the coefficients are 
 \bea
 \begin{split}
& A = 1+2 \epsilon h'(\bar u)\bar u_{\bar y}+\epsilon^2\left(h'(\bar u)^2+m'(\bar u)^2\right)\bar u_{\bar y}^2\\
& B = -2
 \epsilon \left[m'(\bar u)\bar u_{\bar y}+h'(\bar u)\bar u_{\bar x}+\epsilon\left(h'(\bar u)^2+m'(\bar u)^2\right)\bar u_{\bar y}\bar u_{\bar x}\right]
\\
& C= 1+2\epsilon  m'(\bar u)\bar u_{\bar x}
 + \epsilon^2\left(h'(\bar u)^2+m'(\bar u)^2\right)\bar u_{\bar x}^2 \\
& D =  -\epsilon
 \left(m''(\bar u)\bar u_{\bar x}+h''(\bar u)\bar u_{\bar y}\right) \left( \bar u_{\bar x}^2 + \bar u_{\bar y}^2\right) \\
  & E= \left(1+\epsilon(m'(\bar u)\bar u_{\bar x}+h'(\bar u)\bar u_{\bar y})\right)^2
 \end{split}
 \eea
It is obvious that via the  transformations of \eqref{case1}, we can construct  various  maps between   the single heat equation \eqref{heat} and the nonlinear  equations of  the class of  \eqref{nonlinear1}. And any solution of the heat equation \eqref{heat} would generate solutions of nonlinear equations. 
\par 
For the reader to see more clearly how a solution to  nonlinear problems can be constructed via  equivalence transformations we consider the following problem.\\
 {\bf Example 2:}
A subclass of nonlinear diffusion equations \eqref{nonlinear1} can be considered by taking  $m(u)=u^n,\ h(u)=u^r$ in the previous example, where $n,\ r \in \mathbb{R}$.  
\be
\label{un}
\begin{split}
& \left[ 1+2 \epsilon r u^{r-1} u_y + + \epsilon^2 (r^2 u^{2(r-1)}+n^2 u^{2(n-1)}) u_y^2\right] u_{xx} 
 -2 \epsilon \left[ n u^{n-1} u_y \right. \\
 &  \left. +r u^{r-1} u_x + \epsilon (r^2 u^{2(r-1)}+n^2 u^{2(n-1)})u_x u_y\right] u_{xy} + \left[ 1+ 2 \epsilon n u^{n-1} u_x \right. \\
 & \left.  + \epsilon^2 (r^2 u^{2(r-1)}+n^2 u^{2(n-1)})u_x^2 \right] u_{yy} 
  -\epsilon (u_x^2+u_y^2) (r(r-1) u^{r-2} u_y \\
  & + n(n-1) u^{n-2} u_x )  = (1+\epsilon (n u^{n-1} u_x+ r u^{r-1} u_y))^2.
\end{split}
\ee
Here $\bar {(*)} $'s are omitted for the simplicity. And an implicit solution for \eqref{un} can be written as 
\bes
 u-  \bar F \sin{( \mu ( x +\epsilon u^n ))} \sin{( \eta (y+ \epsilon u^r)} e^{-(\mu^2+\eta^2) t}=0
\ees
from a solution like $ 
 u= F \sin{( \mu x )} \sin{( \eta y)} e^{-(\mu^2+\eta^2)t}$ 
 of the heat equation \eqref{heat} by applying the equivalence transformations; $x= \bar x+ \epsilon\bar u^n,\ y=\bar y+ \epsilon \bar u^r ,\ t=\bar t,\ u=\bar u$.   Here $F$ and $\bar F$ are some constants.\\
 {\bf Example 3}: For the reader to follow the procedure better and  check the calculations, let us  consider another particular member of the nonlinear equation \eqref{nonlinear1} by taking $m(u)=u^2$ and $h(u)=0$, where we could simply write
 \be
 \label{example3}
\begin{split}
 &(1+\bar u^2\bar u_{\bar y}^2)\bar u_{\bar x \bar x}
 -2(1+ \bar u\bar u_{\bar x})\bar u\bar u_{\bar y}\bar u_{\bar x\bar y}
  +(1+\bar u\bar u_{\bar x})^2
 \bar u_{\bar y\bar y} 
-(\bar u_{\bar y}^2+\bar u_{\bar x}^2)\bar u_{\bar x}
 \\
&\qquad  =(1+\bar u\bar u_{\bar x})^2\bar u_{\bar t}.
 \end{split}
 \ee
The reader can easily see, a  much more easier function $  u=x+ y^2+2  t$ which satisfies the linear heat equation would generate     
  \bes \bar u=1 \mp \sqrt{1-2( \bar
x+\bar y^2 +2 \bar t))}
\ees
as a solution to the nonlinear equation \eqref{example3} under the transformations $x= \bar x+ \frac{1}{2}\bar u^2,\ y=\bar y,\ t=\bar t,\ u=\bar u$.  \\
 {\bf Example 4:} As another example by 
  choosing  $m(u)=\sin u,\ h(u)=\cos u$ the nonlinear equation would become
  \bea
\begin{split}
&(1+ \bar u_{\bar y}^2-2 \bar u_{\bar y}\sin \bar u )\bar u_{\bar x\bar x}+(1+ \bar u_{\bar x}^2+2\bar u_{\bar x}\cos \bar u )\bar u_{\bar y\bar y}\\
&-2(\bar u_{\bar y}\cos \bar u - \bar
u_{\bar x}\sin \bar u+ \bar u_{\bar x}\bar u_{\bar y}) \bar
u_{\bar x\bar y}+ (\bar u_{\bar y}\cos \bar u +\bar u_{\bar x}\sin
\bar u )(
\bar u_{\bar y}^2+\bar u_{\bar x}^2)\\
&-(1- \bar u_{\bar y} \sin \bar u +\bar u_{\bar x} \cos \bar u )^2 \bar u_{\bar t}=0.
\end{split}
\eea 
can be transformed into the heat equation \eqref{heat} by the equivalence transformations: 
$\bar x=x- sin u,\ \bar y=y- cos u,\ \bar u=u$. Any solution of the linear heat equation would generate a solution of the equation. 
 \subsection{On the subgroup:  $\xi^1= m(u,y),\ \xi^2=\xi^3=\eta=0$}
 Similarly to the previous subsection, transformations related to a subgroup generated by the infinitesimal generators    $\xi^1=m(u,y),\ \xi^2=\xi^3=\eta=0$ are considered. The corresponding  equivalence transformations of the subgroup are explicitly determined as: 
  \beas
 \begin{split}
 \bar x&=x-\epsilon m(u,y),\quad \bar y=y,\quad \bar t=t,\quad \bar u=u,\\
 \bar u_x&=\frac{u_x}{1-\epsilon m_{u}(u,y)u_x}, \quad \bar u_y=\frac{u_y+\epsilon m_{y}(u,y)u_x }{1-\epsilon m_{u}(u,y)u_x},\quad
  \bar u_t=
\frac{u_t}{1-\epsilon m_{u}(u,y)u_x},\\
 \bar f&=f-\frac{ \epsilon (m_{y}(u,y)+m_{u}(u,y)u_y)g} {1-\epsilon m_{u}(u,y)u_x}  ,\quad
  \bar g=\frac{g}{1-\epsilon m_{u}(u,y)u_x}
 \end{split}
 \eeas
 where subscripts indicate partial differentiation with respect to
the corresponding variable.
 \\
{\bf Example 5:}
 For the simplicity, let us again seek the nonlinear diffusion equations can be mapped onto the linear heat equation \eqref{heat}. For $f=u_x$ and $g=u_y$,  
 $\bar f$ and $\bar g$   turn to be:
  \bes \bar f=\frac{\bar u_{\bar x}-\epsilon \left( m_{\bar
y}+m_{\bar u}\bar u_{\bar y}\right)\left(\bar u_{\bar y}-\epsilon
m_{\bar y}\bar u_{\bar x}\right)}{1+\epsilon m_{\bar u}\bar u_{\bar
x}},\qquad \bar g=\bar u_{\bar y}-\epsilon m_{\bar y}\bar u_{\bar x}
\ees where $m=m(\bar u,\bar y)$. Thus the following class are the class of nonlinear equations that mapped onto the heat equation: 
\bes
 A \ \bar u_{\bar x\bar x}+ B \ \bar u_{\bar x \bar y} +C \ \bar u_{\bar y\bar y} + D = E \ \bar u_{\bar t}\ \longrightarrow
 \ u_{xx}+u_{yy}=u_t
 \ees
where the coefficients are
 \bea
\label{nonlinear2}
\begin{split}
& A = 1+\epsilon^2 m_{\bar y}^2+3\epsilon^2 m_{\bar u}m_{\bar y} \bar u_{\bar y}+2 \epsilon^2 m_{\bar u}^2 \bar u_{\bar y}^2\\
& B = \epsilon\left[\epsilon^2 m_{\bar u}^2 \bar u_{\bar x}^2( m_{\bar u}\bar u_{\bar y}+ m_{\bar y})
-2 \epsilon m_{\bar u}^2 \bar u_{\bar x}\bar u_{\bar y}- \epsilon m_{\bar u} m_{\bar y}\bar u_{\bar x}-3  m_{\bar u}\bar u_{\bar y} -2 m_{\bar y}\right]\\
& C = 1+\epsilon m_{\bar u}\bar u_{\bar x}-\epsilon^2 m_{\bar u}^2\bar u_{\bar x}^2-\epsilon^3 m_{\bar u}^3 \bar u_{\bar x}^3\\
& D = \epsilon \left[2 \epsilon^2 m_{\bar u} m_{\bar y}
m_{\bar u\bar y}-\epsilon^2 m_{\bar u}^2 m_{\bar y\bar y}-m_{\bar u\bar u}-\epsilon^2 m_{\bar y}^2
 m_{\bar u\bar u}\right]\bar u_{x}^3 +\epsilon^2\left[-2 m_{\bar y\bar y}m_{\bar u} \right. \\
 &\quad  \left. +2 m_{\bar u\bar y}m_{\bar y}-3 m_{\bar u\bar y}m_{\bar u}
 \bar u_{\bar y}+3 m_{\bar u\bar u}m_{\bar y}\bar u_{\bar y}\right]\bar u_{x}^2 
 -2 \epsilon m_{\bar u\bar u}\bar u_{\bar x}\bar u_{\bar y}^2\\
 & \quad -3\epsilon m_{\bar u\bar y}\bar u_{\bar x}\bar u_{\bar y}-
 \epsilon m_{\bar y\bar y}\bar u_{\bar x}\\
 & E = (1+\epsilon m_{\bar u}\bar u_{\bar x})^2
 \end{split}
 \eea
Any solution $
 \phi(x,y,t,u)=0$ of the heat equation \eqref{heat} generates hereby the solution $\phi(\bar
x+\epsilon m(\bar u,\bar y),\bar y,\bar t,\bar u)=0$ for the nonlinear equation \eqref{nonlinear1} with the coefficients \eqref{nonlinear2}. 
 \subsection{On the subgroup: $\xi^2=m(u,x),\ \xi^1=\xi^3=\eta=0$}
 Completely similar to the previous subsections, the infinitesimal generators  on the vector field of \eqref{prolonged}  
 yield the transformations
 \bea
 \label{case3}
 \begin{split}
 \bar x&=x,\quad \bar y=y-\epsilon m(u,x),\quad
  \bar t=t,\quad \bar u=u,\\
 \bar u_x&=\frac{u_x+\epsilon m_x(u,x)u_y}{1-\epsilon m_{u}(u,x)u_y}, \quad
  \bar u_y=\frac{u_y }{1-\epsilon m_{u}(u,x)u_y},\quad
  \bar u_t=\frac{u_t}{1-\epsilon m_{u}(u,x)u_y},\\
 \bar f&=\frac{ f} {1-\epsilon m_{u}(u,x)u_y}  ,\quad
 \bar g=g-\frac{\epsilon(m_x(u,x)+\epsilon m_u(u,x)u_x)g}{1-\epsilon m_{u}(u,x)u_y}.
 \end{split}
 \eea
 As the procedure is already given in detail in the previous subsections with many examples, here we will  only consider the following  example which is a little different than the previous ones. \\
{\bf Example 6:}
 Even though the first order PDE, $$u_x+u_y-u_t=0$$ is not a diffusion equation, we may still examine it under the  equivalence
 transformations generated for the diffusion equation \eqref{main} by  taking  $f=g=u$. The following class of  nonlinear equations are mapped onto the given constant coefficient PDE under the transformation group \eqref{case3}:  
\be\label{firstorder}
\bar u_{\bar x}+(1-\epsilon m_{\bar x}(\bar u,\bar x))\bar
u_{\bar y}-\bar u_{\bar t}=0.
\ee
And a solution to the equation \eqref{firstorder} can be written as
$$ \bar
u-\psi(\bar t+\bar x,\ \bar y+\epsilon m(\bar u,\bar x)-\bar x)=0 $$
from the general solution of the linear equation
$u=\psi(t+x,y-x)$. Here we should warn the reader that even though the solution of the linear equation is its general solution, the transformed solution is not the general solution to the nonlinear equation.
\subsection{On the subgroup: $\xi^1=\xi^2=\xi^3=0, \ \eta=m(x,y,t)$ }
As the last case, here we shall consider  the equivalence
transformations of linear or nonlinear equations which do not map between each other, but map one into another  with different
coefficient functions. We can construct many of such transformations, but as an example  here will investigate one subgroup by taking the infinitesimal generators  drive the following
 equivalence transformations
 \beas
 \begin{split}
 \bar x&=x,\quad \bar y=y,\quad \bar t=t,\quad \bar u=u+\epsilon m(x,y,t),\\
 \bar u_x&=u_x+\epsilon m_x(x,y,t), \quad \bar u_y=u_y+\epsilon m_{y}(x,y,t),\quad  \bar u_t=u_t+\epsilon m_{t}(x,y,t),\\
 \bar f&= f+\epsilon \int m_{t}(x,y,t)dx  ,\qquad  \bar g=g.
 \end{split}
 \eeas
 Under such transformations $f_x+g_y-u_t=0$ is mapped onto
\bes
 \left[f+\epsilon \int m_{t}(x,y,t)dx\right]_x+g_y-
u_t=\epsilon m_t(x,y,t). \ees
\par
One can say by considering such transformations we may address the nonhomegeneous equations to homogeneous equations. \par
It is obviously clear that one  investigate more subgroups of the general equivalence groups. And for each subgroup, various different maps between linear and nonlinear diffusion equations can be examined. But we hope the subgroups and examples examined here will suffice to give some idea to the general structure on the maps between nonlinear diffusion equations and  linear equations.   
  \section{Differential Invariants}
Differential Invariants of Lie groups of continuous transformations
play important role in mathematical modelling, differential geometry
and nonlinear field equations. In recent years, differential
invariants admitting equivalence transformations have been mostly
studied for constructing  maps between linear and nonlinear
differential equations. The reader may look at \cite{tracina2004invariants} and
\cite{torrisi2004linearization} for the application of differential invariants to the
linearization problem for nonlinear wave equation.  The
linearization problem via differential invariants for one
dimensional diffusion equation was investigated in \cite{torrisi2005second,gandarias2007some,
ibragimov2007differential,torrisi2011exact}. Recently the problem is
studied for third order evolution equation by Tsaousi et al. in
\cite{tsaousi2015differential}. Like the equivalence transformation for the general class of diffusion equation \eqref{main} has not been studied yet its differential invariants have  not been considered in any research as well. \par
 One can understand that determining the differential invariants for the complete group of equivalence transformations for the  equation \eqref{main} is almost impossible. Thus  we will here investigate the
differential invariants  admitting the special
subgroups of equivalence transformations represented in section 3.1 and discuss about the results.
\subsection {Differential Invariants  for the Subgroup: $ \xi^1=m(u), \xi^2=h(u)$}
The subgroup which are  examined in Section 3.1 for
the class of  (2+1) dimensional diffusion equation
(1) is infinite dimensional and spanned by  the vector fields \bea
\label{vmvh}
\begin{split}
V_m &= m(u) \frac{\partial }{{\partial x }} +m'(u)\left[u_x^2 \frac{\partial}{\partial {u_x}} +u_xu_y \frac{\partial}{\partial {u_y}}+u_xu_t \frac{\partial}{\partial {u_t}}-u_yg  \frac{\partial }{{\partial f }}+ u_xg \frac{\partial }{{\partial g }}\right],\\
V_h &= h(u) \frac{\partial }{{\partial y }} +h'(u)\left[u_xu_y \frac{\partial}{\partial {u_x}} +u_y^2 \frac{\partial}{\partial {u_y}}+u_yu_t \frac{\partial}{\partial {u_t}}+u_yf  \frac{\partial }{{\partial f }}-u_xf \frac{\partial }{{\partial g }}\right]
\end{split}
  \eea
where
  $$[V_m,V_h]=V_m(V_h)-V_h(V_m)=0.$$
  \begin{definition} A function $$J=J(x,y,t,u,u_x,u_y,u_t,f,g)$$ is called the  invariant of order zero of the (2+1) dimensional diffusion equation  \eqref{main}, if it is
    invariant under the equivalence groups $V_m$ and $ V_h$ given by \eqref{vmvh}.
    \end{definition}
Invariance conditions  $V_m(J)=0$ and $V_h(J)=0$ yield the invariant function of order zero to be
  $$J=J(t,u,\xi_1,\xi_2,\xi_3)$$
  where \be
  \label{inv0}
  \xi_1=\frac{u_y}{u_x},\quad \xi_2=\frac{u_t}{u_x},\quad
  \xi_3=f\frac{1}{\xi_1}+g.
  \ee
One can easily see that these differential invariants are consistent with the results \eqref{case1} obtained by the  direct integration method
  $$\frac{u_y}{u_x}=\frac{\bar u_{\bar y}}{\bar u_{\bar x}},\qquad \frac{u_t}{u_x}=\frac{\bar u_{\bar t}}{\bar u_{\bar
  x}}, \qquad f\frac{u_x}{u_y}+g=\bar f\frac{\bar u_x}{\bar u_y}+\bar g.$$
 \begin{definition}
 A function   \bes
J(x,y,t,u,u_x,u_y,u_t,f,g,f_x,f_y,f_t,f_u,f_{u_x},f_{u_y},g_x,g_y,g_t,g_u,g_{u_x},g_{u_y})
\ees
is called the invariant of order one for the diffusion equation \eqref{main}, if it is invariant under the equivalence groups related to the prolonged vector fields $\tilde V_m$ and $\tilde V_h$.
 \end{definition}
Computation of  the differential invariants of order one  for this subgroup too complicated as they involve two free functions. For the simplicity, here we
will examine the procedure finding the first order invariants by
considering the subgroup in which $h(u)=0$.
\par 
 In addition to $V_m$ given by
\eqref{vmvh}, substituting the additional components of the vector field \eqref{prolonged} which determined by \eqref{gen2} yields the
prolongation vector $\tilde {V}_m$ to be
 \bea
\begin{split}
\widetilde{V}_m &=
-\left[m'(u)\left(f_x+u_yg_u\right)+m''(u)\left(u_x^2f_{u_x}+u_xu_yf_{u_y}\right)\right]\frac{\partial
}{{\partial f_u }}\\
&- \left[m'(u)\left(g_x-u_x g_u\right)+m''(u)\left(u_x^2g_{u_x}+u_xu_y g_{u_y}-u_x g\right)\right] \frac{\partial}{\partial g_u}\\
&-m'(u)u_yg_x \frac{\partial}{\partial f_x}-m'(u)u_yg_y
\frac{\partial}{\partial f_y}-m'(u)u_yg_t \frac{\partial}{\partial f_t}+m'(u)u_xg_x \frac{\partial}{\partial g_x}\\
&+ m'(u) u_xg_y \frac{\partial }{\partial g_y}+m'(u) u_xg_t
\frac{\partial }{\partial g_t}-m'(u)\left(2
u_xf_{u_x}+u_yf_{u_y}+u_yg_{u_x}\right) \frac{\partial}{\partial
f_{u_x}}\\
&-m'(u)\left( u_xf_{u_y}+u_yg_{u_y}+g\right)\frac{\partial}{\partial
f_{u_y}} -m'(u)\left( u_xg_{u_x}+u_yg_{u_y}-g\right)
\frac{\partial}{\partial g_{u_x}}.
\end{split}
  \eea
  Since $m(u)$ is an arbitrary function, we apply the invariant test $$V_m(J)=0,\qquad
  \widetilde{V}_m(J)=0$$  and obtain the  invariant function of order one as a function depending on 15 invariants
    $$J(y,t,u,\xi_i),\qquad i=1,2,...13$$
    where
    \bea
  \begin{split}
& \xi_1=\frac{u_y}{u_x},\qquad \xi_2=\frac{u_t}{u_x},\qquad
\xi_3=\frac{g}{u_x},\qquad
 \xi_4=f+\xi_1 g,\\
& \xi_5=g_x+\frac{1}{\xi_1}f_x,\qquad
\xi_6=\frac{g_y}{g_x}\frac{1}{\xi_1},\qquad
\xi_7=f_y-\frac{g_y}{g_x}f_x,\\&
\xi_8=\frac{g_t}{g_x}\frac{1}{\xi_1},\qquad
 \xi_9=f_t-\frac{g_t}{g_x}f_x,\qquad  \xi_{10}=g_{u_y}, \\
&
\xi_{11}=-(\xi_1\xi_{10}+\xi_3)f_x+f_{u_y}g_x \xi_1,\qquad \xi_{12}=(\xi_3-\xi_1g_{u_y})f_x+\xi_1g_xg_{u_x},\\
&\xi_{13}=\left[\left(g_{u_y}f_{x}- g_{u_x}g_{x}\right)f_x+\left(f_{u_x}g_{x}-f_{u_y}f_x\right)g_{x}\right]\xi_1^2.
  \end{split}
  \eea
 A symbolic software is used to compute the differential invariants. Differential invariants related to the prolonged vector field $\tilde V_h$ can also be obtained but because the procedure is the same we neglect that part. 
  \section{Conclusion and Remarks}
  In this work we considered the equivalence transformations for a general (2+1) dimensional diffusion equation with no restriction on the functional dependencies of  free functions. The goal was to construct the most general infinitesimal generators for the transformations associated  equivalence group and investigate the structure of admissible transformations between linear and nonlinear equations. We showed that such transformations were only possible when the transformed independent variables involve  the dependent variable. Similar analysis can either be applied to some smaller classes directly or the results obtained here can be applied by appropriate restrictions.

  Second, we considered some subgroups and chose some particular transformations to generate maps between linear and nonlinear equations and we determined the class of nonlinear diffusion equations that can be mapped onto the linear heat equation. We have not interested in algebraic structure and the classification problem. Classification for some members of the diffusion equation can be a subject of another study.

  Third, in the last section we investigated the  differential invariants of order zero and of order one for a subgroup which was examined in the previous section. And were able to show that the zeroth order differential invariants were compatible with the results we obtained by direct integration method.  The determination of the differential invariants can easily be extended by taking some other particular subgroups by running the similar calculations.

    \bibliographystyle{unsrt}
  \bibliography{nonlineardiffusioneq}{}

  \end{document}